\documentclass[runningheads]{llncs}
\usepackage{amsmath}
\usepackage{array,xspace,multirow,hhline,graphicx,xcolor,tikz,colortbl,tabularx,amsmath,amssymb,amsfonts}
\usetikzlibrary{shapes}
\usetikzlibrary{arrows}
\usetikzlibrary{decorations.pathreplacing,calligraphy}

\usepackage{caption,tabularx,booktabs}

 \usepackage[numbers]{natbib}

 \makeatletter 
 \renewcommand\@biblabel[1]{#1} 
 \makeatother
 
\usepackage{graphicx}
\usepackage{caption}
\usepackage{subcaption}
\usepackage{collcell}
\usepackage{hhline}
\usepackage{pgf}
\usepackage{pgfplots}

\usepackage{booktabs}

\usepackage{wrapfig}
\newcommand\eat[1]{}
\usepackage{pgfplots}
\usepackage{bbm}
\usepackage{environ}
\NewEnviron{prob}[1]{%
\begin{center}\fbox{\parbox{4.5in}{%
    {\centering\scshape #1\par}%
    \parskip=1ex
    \everypar{\hangindent=1em}%
    \BODY
}}\end{center}}
\usepackage{verbatim,ifthen}
\usepackage{pifont}
\usepackage{ifthen}
\usepackage{pgflibraryshapes}

\usepackage{nicefrac}
\usetikzlibrary{arrows}
\usepackage{ltxtable}
\usepackage{longtable}
\usepackage[normalem]{ulem}
	\usepackage{varioref}

\usepackage{calc}
\newsavebox\CBox
\newcommand\hcancel[2][0.5pt]{%
  \ifmmode\sbox\CBox{$#2$}\else\sbox\CBox{#2}\fi%
  \makebox[0pt][l]{\usebox\CBox}%
  \rule[0.5\ht\CBox-#1/2]{\wd\CBox}{#1}}

\tikzset{
  jumpdot/.style={mark=*,solid},
  excl/.append style={jumpdot,fill=white},
  incl/.append style={jumpdot,fill=black},
  rexcl/.append style={jumpdot,color=red,fill=white},
  rincl/.append style={jumpdot,fill=black,color=red},
}
\newcommand{\stkout}[1]{\ifmmode\text{\sout{\ensuremath{#1}}}\else\sout{#1}\fi}

\usepackage[ruled]{algorithm2e} 

\SetAlFnt{\small}
\SetAlCapFnt{\small}
\SetAlCapNameFnt{\small}
\SetAlCapHSkip{0pt}
\IncMargin{-\parindent}

	\newcommand{\pref}{\succeq\xspace}

\usepackage{mathtools}

\definecolor{gray(x11gray)}{rgb}{0.75, 0.75, 0.75}
\def\colorModel{rgb} 

\newcommand\ColCell[1]{
  \pgfmathparse{#1<0.5?1:0}  
    \ifnum\pgfmathresult=0\relax\color{white}\fi
  \pgfmathsetmacro\compA{1-#1}      
  \pgfmathsetmacro\compB{1-#1/1.5} 
  \pgfmathsetmacro\compC{1}      
  \edef\x{\noexpand\centering\noexpand\cellcolor[\colorModel]{\compA,\compB,\compC}}\x #1
  } 
\newcolumntype{E}{>{\collectcell\ColCell}m{0.5cm}<{\endcollectcell}}  

\begin{document}
\title{Obvious Manipulability of Voting Rules}
%
%


\author{Haris Aziz\inst{1} \and
Alexander Lam\inst{1}
}
\authorrunning{H. Aziz and A. Lam}

\institute{UNSW Sydney\\
\email{\{haris.aziz,alexander.lam1\}@unsw.edu.au}
}

\maketitle              
\begin{abstract}
  The Gibbard-Satterthwaite theorem states that no unanimous and non-dictatorial voting rule is strategyproof. We revisit voting rules and consider a weaker notion of strategyproofness called not obvious manipulability that was proposed by Troyan and Morrill (2020). We identify several classes of voting rules
that satisfy this notion. We also show that several voting rules including $k$-approval fail to satisfy this property. We characterize conditions under which voting rules are obviously manipulable. 
One of our insights is that certain rules are obviously manipulable when the number of alternatives is relatively large compared to the number of voters. In contrast to the Gibbard-Satterthwaite theorem, many of the rules we examined are not obviously manipulable. This reflects the relatively easier satisfiability of the notion and the zero information assumption of not obvious manipulability, as opposed to the perfect information assumption of strategyproofness. We also present algorithmic results for computing obvious manipulations and report on experiments.

\keywords{Social choice \and  voting \and manipulation \and strategyproofness.}
\end{abstract}

\section{Introduction}

Throughout history, voting has been used as a means of making public decisions based on the citizens' preferences. 
The ancient Greeks would give a show of hands to disclose their most preferred public official, and the winner of the election was chosen as the official with the most first preferences \citep{Chis11a}; such a voting system is called the \textit{plurality vote}. Many other voting systems have been developed over time, such as the Borda Count, developed by Jean-Charles de Borda in 1770. The Borda Count gives each candidate a score based on their position in the voters' preference orders. This system was opposed by Marquis de Condorcet, who instead preferred the Condorcet method, which elects the candidate that wins the majority of pairwise head-to-head elections against the other candidates \citep{Blac58a}. However, voting systems are not just used in politics; voting theory is frequently used and studied in artificial intelligence to aggregate the preferences of multiple agents into a single decision.

The studies of electoral systems in social choice theory have been wrought with negative results. 
Arrow's impossibility theorem \citep{Arro50a} showed that there exists no voting system with three reasonable requirements. In a similar vein, the Gibbard-Sattherthwaite theorem \citep{Gibb73a,Satt75a} states that when there are at least three alternatives, every unanimous voting rule is either dictatorial, meaning only one voter's preferences are taken into account, or prone to manipulative voting, meaning a voter can give an untruthful ballot to gain a more preferred outcome.

Such strategic behaviour is a commonly studied problem in mechanism design and social choice, as many mechanisms sacrifice efficiency or fairness to ensure strategyproofness. 
The original notion of strategyproofness fails to explain the variation we observe in voters' tendency to strategically vote in different electoral systems. This has motivated research toward alternative concepts  of strategyproofness that may be able to capture such variations.
One such notion is \textit{not obvious manipulability}, recently theorized by \citet{TrMo20a}. Whilst strategyproofness assumes agents have complete information over other agent preferences and the mechanism operation, not obvious manipulability assumes agents are `cognitively limited' and lack such information. As such, they are only aware of the possible range of outcomes that can result from each mechanism interaction. Put simply, a mechanism satisfies \textit{not obvious manipulability (NOM)} if no agent can improve its best case or worst case outcome under any manipulation. A mechanism is \textit{obviously manipulable (OM)} if either an agent's best case or worst case outcome can be improved by some untruthful interaction.

The assumptions made for \textit{not obvious manipulability} are suitable when applied to voting rules, as ballots are commonly hidden from the voters, restricting their ability to compute a desirable manipulation. In this paper, we explore which voting rules are obviously manipulable, and if so, what the conditions are for obvious manipulability.


\paragraph{Contributions}

Our main contribution is to apply the concept of obvious manipulations to the case of voting rules for the first time. We study which voting rules are obviously manipulable, and what conditions are required for obvious manipulability. Whilst many classes of voting rules including Condorcet extensions and strict positional scoring rules with weakly diminishing differences are not obviously manipulable, we show that certain voting rules, including $k$-approval, are obviously manipulable. We also characterize the conditions under which positional scoring rules are obviously manipulable in the best case. For the class of $k$-approval voting rules, we characterize the conditions under which the rules are obviously manipulable. Many of our results apply to large classes of voting rules including positional scoring rules or Condorcet extensions. 
Table~\ref{table:summary} summarizes several of our results. 
\begin{table}[h!]
	\centering
\scalebox{0.8}{
\begin{tabular}{ll}
\toprule
\textbf{NOM}&\textbf{OM}\\
\midrule
\textbf{Does not admit a voter with veto power} & \\
\midrule
\textbf{k-Approval} ($n> \frac{m-2}{m-k}$)  & \textbf{k-Approval} ($n\leq \frac{m-2}{m-k}$)\\
Plurality & \\
\midrule
\textbf{Almost-unanimous} & \\
Condorcet-extension & \\
STV & \\
Plurality with runoff & \\
\midrule
Positional scoring rule ($n>\frac{s_1}{s_1-s_2}+1$) ~~~~~&Positional scoring rule that \\

& admits a voter with veto power (existence)\\
\midrule
Positional scoring rule  with weakly&\\
{diminishing differences} & \\
Borda rule & \\
\bottomrule
\end{tabular}
}
\caption{List of rules and conditions for voting rules to be NOM or OM.}
\label{table:summary}
\end{table}

One of our insights is that certain rules are obviously manipulable when the number of alternatives is relatively large compared to the number of voters. We also look at the problem of checking whether a particular instance of a voting problem admits an obvious manipulation. For the class of positional scoring rules, we provide a general polynomial-time reduction to the well-studied \emph{unweighted coalitional manipulation problem}. As a corollary, we show that the problem of checking the existence of an obvious manipulation is polynomial-time solvable for the $k$-approval rule. Finally, we report on experimental results on the fraction of instances that admit obvious manipulations for the $k$-approval rule.


\section{Related Work}

Our paper belongs to the rich stream of work in social choice on the manipulability of voting rules. The reader is referred to the book by \citet{Tayl05a} that surveys this rich field. A comparison of the susceptibility of voting rules to manipulation has a long history in social choice. For example, one particular approach is to count the relative number of preference profiles under which voting rules are manipulable (see, e.g., \citep{FLS02a}). Another approach is analyzing the maximum amount of expected utility an agent can gain by reporting untruthfully~\citep{Carr11a}.

Our work revolves around the concept of obvious manipulations, which was proposed by \citet{TrMo20a}. This concept was inspired by a paper on `obviously strategyproof mechanisms' by \citet{Li17a}. The latter paper describes the cognitively-limited agent that is only aware of the range of possible outcomes ranging from each report. In the paper, Li then proposes the characterization of `obvious strategyproofness', a strengthening of strategyproofness. A mechanism is defined as obviously strategyproof if each agent's worst case outcome under a truthful report is strictly better than their best case outcome under any untruthful report.
\citet{TrMo20a} studied obvious manipulations in the context of matching problems. In particular, they showed that whereas the Boston mechanism is obviously manipulable, many stable matching mechanisms (including those that are not strategyproof) are not obviously manipulable.

Other, weaker notions of strategyproofness specific to voting rules have been proposed in the literature. \citet{SlWh08a,SlWh14a} considered \emph{safe strategic voting} to represent the coalitional manipulation of scoring rules. Assuming every member of the coalition reports the same ballot, a manipulation is a \emph{safe strategic vote} if it guarantees an outcome which is weakly preferred over truth-telling. {Another notion has also been proposed by \citet{CWX11a}, who state that a ballot \emph{dominates} another ballot if it guarantees a weakly more preferred outcome. The authors define a voting rule as being \emph{immune to dominating manipulations} if there are no ballots that dominate any voter's true preferences, and classify the immunity of certain rules under varying levels of information known by the manipulator. In particular relevance to our paper, they find that certain voting rules such as Condorcet-consistent rules and the Borda count are immune to dominating manipulations under zero information. We remark that immunity to dominating manipulations under zero information is a weaker notion than not obvious manipulability, and thus our work investigates a stronger notion defining a voting rule's resistance to manipulation than some existing notions.} For further discussion on the strategic aspects of voting with partial information, the reader is referred to Chapter 6 and 8 of the book by \citet{Meir18a}, {where similar concepts such as local dominance are discussed}.


In many elections, voters often lack information of other voters' preferences. This has prompted a probabilistic perspective into the manipulability of voting rules, often assuming a uniform distribution over each preference ordering. In 1985, Nitzan showed that in point scoring rules, a manipulation is more likely to succeed as the number of outcomes increases, and the number of voters decreases \citep{Nitz85a}. A similar probabilistic perspective was used by \citet{WiRe10a}. 
Computer scientists have also extensively researched the computational complexity of calculating a manipulative ballot; as the number of voters and outcomes becomes large, it can be computationally infeasible to compute a manipulation if the problem is intractable (see, e.g. \citep{CoWa16a,FaPr10a}). 

\section{Preliminaries}
We consider the standard social choice voting setting $(N,O,\succ)$ that involves a finite set 
$N=\{1,2,\dots,n\}$ of $n$ voters and a finite set $O=\{o_1,o_2,\dots,o_m\}$ of $m$ outcomes. We also assume that $n\geq 3$ and $m\geq 3$. Each voter $i$ has a transitive, complete and reflexive preference ordering $\succ_i$ over the set of outcomes $O$. We denote the preference profile of each voter $i\in N$ as $\succ=(\succ_1,\ldots, \succ_n)$, and use $\mathcal{L}(O)^n$ to denote the set of all such profiles for a given $n$. For a given voter $i\in N$, we use $\succ_{-i}=(\succ_1,\dots,\succ_{i-1},\succ_{i+1},\dots,\succ_n)$ to denote the preference profile of the voters in $N\backslash \{i\}$.
A voting rule $f{: \mathcal{L}(O)^n\rightarrow O}$ is a function that takes as input the preference profile and returns an outcome from $O$.




An outcome $o\in O$ is called a \emph{possible outcome} under a voting rule $f$ if there exists some preference profile $\succ$ such that $f(\succ)=o$.
Since we are considering voting rules that return a single outcome, we will impose tie-breaking over social choice correspondences (voting rules that return more than one outcome) to return a single outcome. Unless specified otherwise, we will assume a fixed tie-break ordering over the outcomes. 

\begin{definition}
	A voting rule $f$ is \emph{manipulable} if there exists some voter $i \in N$, two preference relations  $\succ_i, \succ'_i$ of voter $i$, and a preference profile $\succ_{-i}$ of other voters such that $f(\succ_i',\succ_{-i}) \succ_i f(\succ_i,\succ_{-i})$. Such a manipulation is defined as a \emph{profitable manipulation} for voter $i$. 
A voting rule is \emph{strategyproof (SP)} if it is not manipulable. 
\end{definition}

Under voting rule $f$, a given set of outcomes and a fixed number of voters, we denote by $B_{\succ_i}(\succ_i',f)$ the best possible outcome (under $i$'s preference $\succ_i$) when she reports $\succ_i'$, over all possible preferences of the other voters. We also denote by $W_{\succ_i}(\succ_i', f)$ the worst possible outcome (under $i$'s preference $\succ_i$) when she reports $\succ_i'$, over all possible preferences of the other voters. We now present the central concept used in the paper, which has been adapted from the paper by Troyan and Morrill \citep{TrMo20a} to the field of voting.

\begin{definition}
A voting rule $f$ is {\it not obviously manipulable} (NOM) if for every voter $i$ with truthful preference $\succ_i$ and every profitable manipulation $\succ_i'$, the following two conditions hold:
\begin{eqnarray}
\label{eq:nom-inf}
  W_{\succ_i}(\succ_i,f)\pref_i W_{\succ_i}(\succ_i',f)
\\
\label{eq:nom-sup}
 B_{\succ_i}(\succ_i,f)\pref_i B_{\succ_i}(\succ_i',f). 
\end{eqnarray}
\end{definition}
If either condition does not hold, then we say the voting rule is \textit{obviously manipulable}. Specifically, if $(1)$ does not hold, then we say the voting rule is \textit{worst case obviously manipulable}. Similarly, if $(2)$ does not hold, then we say it is \textit{best case obviously manipulable}.

%

\section{Sufficient Conditions for not being Obviously Manipulable}
In this section, we identify certain conditions that imply not obvious manipulability when satisfied by voting rules.
\begin{definition}
For a given voting rule $f$ and a fixed number of voters $n$ and outcomes $m$, a voter $i$ has \textit{veto power} if there exists a possible outcome $o\in O$ and report $\succ_i$ {such that $f(\succ_i,\succ_{-i})\neq o$ for all $\succ_{-i}$}. 
\end{definition}

Our first result is a sufficient condition for a voting rule being NOM.

\begin{lemma}\label{prop:noveto}
If a voting rule is obviously manipulable, then it must admit a non-dictatorial {voter with veto power}.
	\end{lemma}
		
%
%

{However,} existence of a voter with veto power does not imply obvious manipulability. We will illustrate this later in the paper.

\begin{definition}
A voting rule $f$ is \emph{almost-unanimous} if it returns an outcome $o$ when $o$ is the most preferred outcome for {at least $n-1$ voters}. Almost-unanimity implies unanimity.
\end{definition}

			\begin{theorem}\label{prop:almostunan}
				For $n\geq 3$, no almost-unanimous voting rule is obviously manipulable.
				\end{theorem}
				\begin{proof}
					Note that an almost-unanimous voting rule is not dictatorial. {By definition}, a rule that is almost-unanimous {cannot} admit a voter with veto power. Hence it follows from Lemma~\ref{prop:noveto} that for $n\geq 3$, no almost-unanimous voting rule is obviously manipulable. \qed
\end{proof}

		\begin{corollary}\label{corol:NOM}
			Any majoritarian (Condorcet extension rule) is NOM. 
			\end{corollary}
%

Similarly, Theorem~\ref{prop:almostunan} applies to several voting rules including STV~\citep{Tide95a} and Plurality with runoff~\citep{Niou01a} that are almost-unanimous.
			
\begin{corollary}
STV and Plurality with runoff are NOM.
\end{corollary}
We have shown that many voting rules are not obviously manipulable, so we question whether there are any obviously manipulable voting rules. We next investigate positional scoring rules.


\section{Positional Scoring Rules}
In this section, we consider positional scoring rules, a major class of voting rules which assigns points to candidates based on voter preferences and chooses the candidate with the highest score. A formal definition of a positional scoring rule is given below.
\begin{definition}
A \textit{positional scoring rule} assigns a score to each outcome using the score vector $w=(s_1,s_2,\dots,s_m)$, where $s_i\geq s_{i+1} \forall i\in \{1,2,\dots,m-1\}$ and $\exists i\in \{1,2,\dots,m-1\}: s_i>s_{i+1}$. Each voter gives $s_i$ points to their $i$th most preferred candidate, and the score of a candidate is the total number of points given by all voters.
The candidate with the highest number of points is returned by the rule.
\end{definition}
{Note that this positional scoring rule definition rules out unreasonable, pathological scoring vectors such as $(1,2,3)$.}
Several well-known rules fall in the class of positional scoring rules. For example if $s_i=m-i$ for all $i\in [m]$, the rule is the Borda voting rule. If $s_1=1$ and $s_i=0$ for all $i>1$, the rule is plurality. If $s_m=0$ and $s_i=1$ for all $i<m$, the rule is anti-plurality. 

Next, we identify a sufficient condition for a positional scoring rule to be NOM.

			\begin{theorem}
			A positional scoring rule is NOM if $n>\frac{s_1}{(s_1-s_2)}+1$.
			\end{theorem}
			\begin{proof}
			It is sufficient to show that for $n>\frac{s_1}{(s_1-s_2)}+1$, the rule is almost-unanimous. Any outcome $a$ that is the most preferred by $n-1$ voters has a score of at least $(s_1)(n-1)$. We show that this score is greater than the score of any other candidate. 
			The maximum score any other outcome $b$ can get is by being in the first position of one voter and second position of all other voters so its score is $(s_2)(n-1)+s_1$. The score of $a$ is greater than the maximum score of $b$ if and only if 
			\begin{align*}
			&(s_1)(n-1)> (s_2)(n-1)+s_1\\
			\iff &(n-1)(s_1-s_2)>s_1 \text{\qquad}\\
			\iff&n> \frac{s_1}{(s_1-s_2)}+1.
			\end{align*} \qed
			\end{proof}

{This result suggests that many positional scoring rules are NOM when there are sufficiently many voters, and that scenarios with few voters may be required for a positional scoring rule to be obviously manipulable.}

\subsection{k-Approval}
The $k$-approval rule is a subclass of positional scoring rules that lets voters approve of their $k$ most preferred candidates, or voice their disapproval for their $m-k$ least preferred candidates. It is a scoring rule with weight vector $w=(1,\dots,1,0,\dots,0)$, where there are $k$ ones, $m-k$ zeroes and $0<k<m$.

Note that the $k$-approval rule is the same as the plurality rule when $k=1$, and it is the same as the anti-plurality rule when $m-k=1$. 
\begin{lemma}\label{k1}
The $k$-approval rule (kApp) is obviously manipulable if $n\leq \frac{m-2}{m-k}$. 
\end{lemma}
\begin{proof}
Suppose there are $n$ voters, the number of outcomes $m$ is at least $n(m-k)+2$, voter $i$'s true preferences are
\[\succ_i: o_1\succ_i o_2 \succ_i \dots \succ_i o_{m-1}\succ_i o_m,\]
and the fixed tie-break ordering is
\[\succ_L: o_k \succ_L o_1\succ_L o_2\succ_L \dots \succ_L o_{k+1}\succ_L o_{k+2}\succ_L \dots \succ_L o_{m-1}\succ_L o_m.\]

Under a $k$-approval rule, any voter may disapprove of their $m-k$ least preferred outcomes. Since there are a total of $n(m-k)$ disapprovals and $m\geq n(m-k)+2$, by the pigeonhole principle, there are at least 2 outcomes with zero disapprovals. Therefore the selected outcome must be the tie-break winner of the outcomes with zero disapproval votes, as they are approved by every voter.

Under a truthful ballot $\succ_i$, voter $i$ disapproves of outcomes $\{o_{k+1},\dots,o_{m}\}$, so $W_{\succ_i}(\succ_i,kApp)\notin\{o_{k+1},\dots,o_{m}\}$. We therefore have $W_{\succ_i}(\succ_i,kApp)=o_k$ 
{as at least two outcomes in $\{o_1,\dots,o_{k}\}$ must have zero disapproval votes, and $o_k$ has the highest tie-break priority}.

If voter $i$ instead disapproves of the outcomes in $\{o_{k}\}\cup\{o_{k+1},\dots,o_m\}\backslash \{o_{i'}\}$, where $k+1\leq i' \leq m$, then the worst case outcome satisfies $W_{\succ_i}(\succ'_i,kApp)\succ_i o_{k-1}$, as $o_{i'}$ always loses the tie-break with any outcome from $\{o_1,\dots,o_{k-1}\}$.\\
We therefore have $W_{\succ_i}(\succ'_i,kApp)\succ_iW_{\succ_i}(\succ_i,kApp)$, concluding the proof. \qed
\end{proof}
\begin{lemma}\label{k2}
The $k$-approval rule (kApp) is NOM if 
$n> \frac{m-2}{m-k}$.
\end{lemma}
\begin{proof}
Suppose that there are $n$ voters, $m\geq \frac{kn-1}{n-1}$ outcomes and without loss of generality that voter $i$'s true preferences are
\[\succ_i: o_1\succ_io_2\succ_i\dots\succ_io_m.\]
We note that $m\geq \frac{kn-1}{n-1} \iff n(m-k)\geq m-1$, so there are at least $m-1$ disapproval votes as each of the $n$ voters disapproves of $m-k$ outcomes. We first show that under these conditions, the $k$-approval rule is not best case obviously manipulable. Under $\succ_i$, voter $i$'s best case outcome of $B_{\succ_i}(\succ_i,kApp)=o_1$ is achievable by the voters voting such that $o_1$ has zero disapprovals and each of the other outcomes has at least one disapproval. Since $i$'s best case outcome is his first preference, it cannot be strictly improved by any manipulation.

We next show that in this scenario, the $k$-approval rule is not worst case obviously manipulable. We consider two cases.

\noindent
\textbf{Case 1 ($n(m-k)=m-1)$)}:

\noindent
In this case, there are $m=n(m-k)+1$ outcomes and $n(m-k)$ disapprovals, so by the pigeonhole principle, there must be at least one outcome with zero disapprovals. Under a truthful ballot, voter $i$ disapproves of outcomes $\{o_{k+1},\dots,o_m\}$, so his worst case outcome is $W_{\succ_i}(\succ_i,kApp)=o_k$, achieved by the other voters disapproving of outcomes $\{o_1,\ldots,o_{k-1}\}$. Now under any manipulation, at least one outcome from $\{o_{k+1},\ldots,o_m\}$ must be approved by voter $i$. This results in $W_{\succ_i}(\succ'_i,kApp)\in \{o_{k+1},\ldots,o_m\}$, as the other voters can vote such that every outcome except for voter $i$'s least preferred approved outcome has been disapproved at least once.

\noindent
\textbf{Case 2 ($n(m-k)>m-1)$)}:

\noindent
In this case, there are at least as many disapprovals as outcomes, so it is possible for each outcome to have at least one disapproval. If $W_{\succ_i}(\succ_i,kApp)\in \{o_{k+1},\ldots,o_m\}$, then each outcome must have at least one disapproval as voter $i$ disapproves of outcomes $\{o_{k+1},\dots,o_m\}$. Suppose voter $i$ misreports ballot $\succ'_i$, switching its disapprovals from at least one outcome in $\{o_{k+1},\ldots,o_{n}\}$ to at least one outcome in $\{o_{1},\ldots,o_{k}\}$. His worst case outcome under $\succ'_i$ cannot improve, as the other agents can `negate' the misreport by switching their disapprovals from the appropriate outcomes in $\{o_{1},\ldots,o_{k}\}$ to $\{o_{k+1},\ldots,o_{n}\}$. This must be possible as each outcome originally had at least one disapproval.

If $W_{\succ_i}(\succ_i,kApp)\in \{o_1,\ldots,o_k\}$, then any manipulation $\succ'_i$ by voter $i$ results in $W_{\succ_i}(\succ'_i,kApp)\in \{o_{k+1},\ldots,o_m\}$ as he will have approved of at least one outcome in $\{o_{k+1},\ldots,o_m\}$. Therefore $W_{\succ_i}(\succ_i,kApp)\succ W_{\succ_i}(\succ'_i,kApp)$.

By exhaustion of cases, we have $W_{\succ_i}(\succ_i,kApp)\pref_i W_{\succ_i}(\succ'_i,kApp)$, concluding our proof.\qed	
\end{proof}
\begin{remark}
We note that the obvious manipulability of $k$-approval when $m\geq n(m-k)+2$ and the not obvious manipulability of $k$-approval when $m= n(m-k)+1$ also holds in the case of weighted voters, as the argument relies on the number of outcomes exceeding the total number of disapprovals.
\end{remark}

Based on the two lemmas proved above, we achieve a characterization of the conditions under which the $k$-approval rule is obviously manipulable.

\begin{theorem}\label{theo:kapp}
The $k$-approval rule is obviously manipulable if and only if $n\leq \frac{m-2}{m-k}$.

\end{theorem}
\begin{corollary}\label{corol:NOMM}
The plurality rule is NOM.
\end{corollary}

Since plurality is generally considered to be one of {easiest rules} to manipulate, the corollary above underscores the strength of obvious manipulations. {We give the following intuition for the result on $k$-approval. Suppose a small committee is applying the $k$-approval rule to select a prize winner out of many candidates, and that certain candidates will be approved by every voter. The manipulator may also have a general idea of these candidates conditional on their report. If a fixed tie-break method is used (such as selecting the oldest candidate), the manipulator may disapprove of the oldest candidate who would otherwise win, instead approving a younger candidate who would not be selected regardless.

 }

\subsection{Strict Positional Scoring Rules}

In the previous section, we noted that the $k$-approval rule is obviously manipulable. This may lead to the question of whether the lack of strictly decreasing scoring weights contributes to the obvious manipulability of a positional scoring rule. Hence, we focus on strict positional scoring rules in the following section.

\begin{definition}
A positional scoring rule with weight vector $w=(s_1,s_2,\ldots,s_m)$ is \emph{strict} if $s_i>s_{i+1}$ for all $i\in \{1,2,\ldots,m-1\}$.
\end{definition}

We first note a strict positional scoring rule can be obviously manipulable. 

\begin{lemma}\label{lemma:spsr}
There exists a strict positional scoring rule that can admit a voter with veto power and is obviously manipulable.
\end{lemma}

In the following lemma, we also find that a strict positional scoring rule is not necessarily obviously manipulable if it admits a voter with veto power.
\begin{lemma}\label{yuck}
There exists a class of strict positional scoring rules that can admit a voter with veto power but are NOM.
\end{lemma}

\begin{definition}
A strict positional scoring rule with $w=(s_1,s_2,\ldots,s_m)$ has \emph{diminishing differences} if $s_i-s_{i+1}>s_{i+1}-s_{i+2}$ for all $i\in\{1,2,\ldots,m-2\}$. We say it has \emph{weakly diminishing differences} if $s_i-s_{i+1}\geq s_{i+1}-s_{i+2}$ for all $i\in\{1,2,\ldots,m-2\}$.
\end{definition}
An example of such a rule is the Harmonic-Borda/Dowdall system used in Nauru, which has weight vector $w=(1,1/2,\ldots,1/m)$ \citep{Reil02a}. It is more favourable towards candidates {that are the top preference of many voters}, and {has been described as a scoring rule that ``lies between plurality and the Borda count" \citep{FrGr14a}}.

Next, we prove that a strict positional scoring rule with weakly diminishing differences is NOM.
\begin{theorem}\label{th:sPOS}
A strict positional scoring rule with weakly diminishing differences is NOM.
\end{theorem}

\begin{corollary}
The Borda and Harmonic-Borda/Dowdall rules are NOM.
\end{corollary}
\begin{remark}
Lemma \ref{yuck} exemplifies a class of strict positional scoring rules which do not satisfy weakly diminishing differences but are NOM.
\end{remark}
\subsection{Obvious Manipulability in the Best Case}
Although our previous results focus on worst case obvious manipulability, it is possible for a positional scoring rule to be best case obviously manipulable.
\begin{lemma}\label{BOM}
Assuming $m,n\geq 3$, a positional scoring rule $f$ is best case obviously manipulable if and only if for some $k>1$, the first $k$ elements of the scoring vector are the same and $n \leq \frac{m-2}{m-k}$.
\end{lemma}
Next, we demonstrate a fundamental connection between best case obvious manipulations and worst case obvious manipulations. 
\begin{theorem}\label{BOMWOM}
Assuming $m,n\geq 3$, for any positional scoring rule, if a voter's preference relation $\succ_i$ admits a best case obvious manipulation, then it also admits a worst case obvious manipulation.
\end{theorem}
\begin{proof}
Suppose for some positional scoring rule $f$ that a voter's preference relation $\succ_i$ admits a best case obvious manipulation. From Lemma \ref{BOM}, for some $k>1$, the first $k$ elements of the scoring vector must be the same, and we have $n\leq \frac{m-2}{m-k}$. Consequently, any outcome selected under $f$ must be in the top $k$ outcomes of each voter's report. We say that a voter `approves' his $k$ most preferred outcomes, and `disapproves' of his $m-k$ least preferred outcomes. An outcome cannot be chosen by $f$ if it has a disapproval vote from at least one voter.

We now construct the set of feasible outcomes $O_f$ which can be selected under the voter's preference relation $\succ_i$ and some $\succ_{-i}$. Let $O_v$ be the $m-k$ disapproved outcomes by $i$ under $\succ_i$. Since any outcome with at least one disapproval vote cannot be chosen, no outcome in $O_v$ can be selected. Now consider the set $O\backslash O_v$. Suppose without loss of generality that $O\backslash O_v = \{o_1,\dots,o_k\}$, with tie-break ordering $\succ_L: o_1\succ_L\dots\succ_Lo_k$. Denote $c:=(n-1)(m-k)$ as the number of disapproval votes that the other $n-1$ voters can distribute. For $j\in \{1,\dots,c+1\}$, outcome $o_j$ can be selected if the other voters cast disapproval votes for outcomes $\{o_1,\dots,o_{c+1}\}\backslash \{o_j\}$. Furthermore, outcomes $o_{c+2},\dots,o_k$ cannot be selected, regardless of how the other voters report. Therefore the set of feasible outcomes $O_f$ are the $c+1$ highest tie-breaking ranked outcomes of the set $O\backslash O_v$. Voter $i$'s best case outcome is its most preferred outcome in $O_f$, whilst its worst case outcome is its least preferred outcome in $O_f$.  We denote $o_b:=B_{\succ_i}(\succ_i,f)$ as $i$'s best case outcome, and $o_w:=W_{\succ_i}(\succ_i,f)$ as $i$'s worst case outcome.

We now define the set of feasible outcomes $O_f'$ under any preference report by voter $i$. This is the $n(m-k)+1$ highest tie-break ranked outcomes of $O$. Now suppose $\succ_i$ admits a best case obvious manipulation. There must exist an outcome $o_b'\in O_f'\backslash O_f$ that $i$ prefers over $o_b$. Consider the set $O_v'=\{o_w\}\cup O_f'\backslash (O_f\cup o_b')$. Since $o_w\not\in O_f'\backslash O_f$ and $o_b'\not\in O_f$, we have
\begin{align*}
|O_v'|&=|\{o_w\}|+|O_f'|-|O_f|-|\{o_b'\}|\\
&= 1+n(m-k)+1-(n-1)(m-k)-1-1\\
&=m-k.
\end{align*} 
We now deduce $i$'s worst case outcome $W_{\succ_i}(\succ_i',f)$ under the manipulation $\succ_i'$ where voter $i$ disapproves of all outcomes from $O_v'$. Under $\succ_i'$, every outcome in $O_f'\backslash O_f$ except for $o_b'$ has a disapproval vote and therefore cannot be selected. The outcome $o_b'$ satisfies $o_b'\succ_i o_b$ and therefore cannot be the worst case outcome. Finally, $o_w$ has a disapproval vote, so by elimination, we have $W_{\succ_i}(\succ_i',f)\in O_f\backslash \{o_w\}$. Since $o_w$ is voter $i$'s least preferred outcome in $O_f$, we have $W_{\succ_i}(\succ_i',f)\succ_i o_w$, meaning that $\succ'_i$ is a worst case obvious manipulation. \qed
\end{proof}
\section{Computing Obvious Manipulations}

In the previous parts of the paper, we focussed on understanding the conditions under which a voting rule is obviously manipulable. Next, 
we consider the problem of computing an obvious manipulation for a given problem instance. We present algorithmic results for computing obvious manipulations under positional scoring rules.

\begin{prob}{Obvious Manipulation (OM)}
	
Input: Number of voters $n$, set of outcomes $O=\{o_1,o_2,\dots,o_m\}$, preference relation $\succ_i$ of voter $i$, {tie-break order} $\succ_L$ and voting rule $f$.

Problem: Find a preference relation $\succ_i'$ such that\\ $W_{\succ_i}(\succ_i',f)\succ_i W_{\succ_i}(\succ_i,f)$
or 
$B_{\succ_i}(\succ_i',f)\succ_i B_{\succ_i}(\succ_i,f)$.
\end{prob}

If we only consider the best case manipulation, we refer to the problem as {\sc Best-case Obvious Manipulation (BOM)}. If we only consider the worst case manipulation, we refer to the problem as {\sc Worst-case Obvious Manipulation (WOM)}. 

%
%
%
%
%
%

We present algorithms for the obvious manipulation problems. The algorithms are based on reductions to the Constructive Coalitional Unweighted Manipulation (CCUM) that is well-studied in computational social choice~(see e.g., \citep{XZP+09a,ZPR09a}). 
We now introduce the \textsc{Constructive Coalitional Unweighted Manipulation (CCUM)}. 
\begin{prob}{Constructive Coalitional Unweighted Manipulation (CCUM)}
Input: Voting rule $f$, set of outcomes $O$, distinguished candidate $o\in O$, set of voters $S$ that have already cast their votes and set of voters $T$ that have not cast their votes.

Problem: Is there a way to cast the votes in $T$ such that $o$ wins the election under $f$?
\end{prob}

We show that for any voting rule, there is a polynomial-time algorithm for computing a best case obvious manipulation if \textsc{CCUM} can be solved in polynomial time. 

\begin{lemma}\label{lemma:bomalgo}
	For any voting rule, there is a polynomial-time algorithm for \textsc{BOM} if \textsc{CCUM} can be solved in polynomial time. 
	\end{lemma}
	\begin{proof}
 Denote $o_b:=B_{\succ_i}(\succ_i,f)$. We can compute $o_b$ as follows. We fix the preference $\succ_i$ of voter $i$ and solve \textsc{CCUM} for each possible outcome while keeping all the other voters as manipulators. This can be checked in $|O|$ calls to an algorithm to solve \textsc{CCUM}.
		Next, we find $i$'s best possible outcome if she is allowed to report any other preference. This can be checked by solving \textsc{CCUM} for each possible outcome while keeping all the voters as manipulators. Let $o^*$ be the possible outcome that is most preferred with respect to $\succ_i$. The instance is best case obviously manipulable if and only if $o^*\succ_i o_b$. \qed
		\end{proof}

We then show that for any positional scoring rule, there is a polynomial-time algorithm for \textsc{OM} if \textsc{CCUM} can be solved in polynomial time.

	\begin{lemma}\label{lemma:algo}
		For any positional scoring rule, there is a polynomial-time algorithm for \textsc{WOM} if \textsc{CCUM} can be solved in polynomial time. 
		\end{lemma}
	
	\begin{proof}
		
First we compute the worst case outcome $W_{\succ_i}(\succ_i,f)$ of $i$ when she reports the truth.
This is easily computed by running an algorithm that solves \textsc{CCUM} with $i$'s report being fixed, and checking which outcomes are possible.  
		
We check whether $i$ can improve her worst case outcome by misreporting. 
We denote $o_w:=W_{\succ_i}(\succ_i,f)$ and $O_{bad}:=\{o\in O: o_w \succ_i o\} \cup \{o_w\}$.
We also denote $O_{good}:=A\setminus O_{bad}$.
We want to check whether $i$ can ensure that no outcome from $O_{bad}$ is selected irrespective of how the other voters vote. We define a misreport $\succ_i'$ as follows. In $\succ_i'$, the outcomes of $O_{good}$ are preferred over the outcomes of $O_{bad}$. In $O_{good}$ the outcomes are ordered so that higher (tie-break) priority outcomes come earlier. 
In $O_{bad}$ the outcomes are ordered so that higher priority outcomes come later. We solve \textsc{CCUM} with respect to $\succ_i'$ and check whether some outcome in $O_{bad}$ can be selected. If such an alternative cannot be selected, we return yes. Otherwise we return no. 
\qed
\end{proof}
%

Combining the two lemmas above, we get the following.

\begin{theorem}
	For any positional scoring rule, there is a polynomial-time reduction from solving \textsc{OM} to solving \textsc{CCUM}.
	\end{theorem}

{\citet{CoWa16a} discuss the computational complexity of \textsc{CCUM} for various different voting rules.} In particular, \textsc{CCUM} can be solved in polynomial time for the $k$-approval problem. For example, \citet{ZPR09a}  present a greedy polynomial-time algorithm for computing \textsc{CCUM}.
For the sake of completeness, we explicitly write this algorithm for the $k$-approval rule with a fixed tie-break ordering. The algorithm assigns approved outcomes to the manipulators as follows.
First, it assigns the distinguished outcome as each manipulator's first preference. Each manipulator then approves the $k-1$ outcomes with the lowest scores. If there are more than $k-1$ tied outcomes, the ones with the lowest tie-break priority are selected.


\begin{corollary}
		\textsc{OM} can be solved in polynomial time for $k$-approval.
	\end{corollary}
	
\subsection*{Experimental Results}

Since the $k$-approval rule is obviously manipulable and obvious manipulations can be found in polynomial time, we further investigate these manipulations in an experiment.
Below, we experimentally determine the effects of $k$, $m$ and $n$ on the proportion of obviously manipulable voter preferences under the $k$-approval rule. Assuming a fixed tie-break ordering, we generate $1$ million randomly permuted voter preference orderings and determine what proportion of these orderings admit an OM for a given set of parameters. It suffices to simply consider individual preference orderings as the best- and worst-case outcomes (and therefore obvious manipulability) for an agent's preference relation are over all possible preferences of the other agents. Note that from Theorem \ref{BOMWOM}, the set of WOM-admitting preference orderings is the same as the set of OM-admitting preference orderings.
\noindent
\textbf{Effect of $n$:}
Figure \ref{fig1} depicts the results from our experiments determining the effect of the number of voters $n$ on the proportion of obviously manipulable preference orderings. The downwards trend is concurrent with the existing theory that the proportion of individually manipulable voting profiles approaches zero as the number of voters tends to infinity \citep{Pele79a}. A significantly lower proportion of preference orderings admit a BOM than those that admit a WOM.
These trends are consistent for other values of $m$ and $k$, though other figures are omitted due to space restrictions.

\noindent
\textbf{Effect of $m$ and $m-k$:}
In Figure \ref{fig2}, we show heat maps of the proportion of OM-admitting preferences for $m\in\{21,\dots,30\}$ and $m-k$ values for which the preference profile is obviously manipulable. It is more appropriate to consider the number of disapprovals $m-k$ than the number of approvals $k$, as the impact of $k$ is relative to its difference from the number of outcomes. For example, it is better to compare $m=21, k=20$ with $m=30, k=29$ than with $m=30, k=20$. For a fixed number of disapprovals, the proportion of OM-admitting preferences increases with the number of outcomes. This is likely because a lower proportion of the outcomes can be `blocked' by the other voters under the worst case outcome. The proportion increases steadily then rapidly decreases as the number of disapprovals increases, suggesting that an intermediary number of disapprovals increases individual manipulative power in comparison to the manipulative coalition of the other voters.

\begin{center}
\pgfplotsset{width=0.6\linewidth,compat=1.9}
\begin{figure}[h!]
	\centering

       \begin{tikzpicture}
\begin{axis}[
    xlabel={Number of voters $n$},
    ylabel={Proportion of preferences},
    xmin=3, xmax=14,
    ymin=0, ymax=0.8,
    xtick={3,4,5,6,7,8,9,10,11,12,13,14},
    ytick={0,0.2,0.4,0.6,0.8},
    legend pos=north east,
    ymajorgrids=true,
    grid style=dashed,
]

\addplot[
    color=blue,
    mark=square,
    ]
    coordinates {
    (3,0.55)(4,0.53)(5,0.5)(6,0.46)(7,0.41)(8,0.36)(9,0.3)(10,0.24)(11,0.18)(12,0.12)(13,0.06)(14,0)
    };
    \legend{Proportion admitting WOM, Proportion admitting BOM}
    
    \addplot[
    color=red,
    mark=square,
    ]
    coordinates {
    (3,0.18)(4,0.13)(5,0.1)(6,0.08)(7,0.06)(8,0.04)(9,0.03)(10,0.024)(11,0.017)(12,0.01)(13,0.005)(14,0)
    };
    
\end{axis}
\end{tikzpicture}
%
    \caption{Effect of $n$ on proportion of preferences that admit WOM and BOM ($k=14, m=15$)}\label{fig1}
\end{figure}
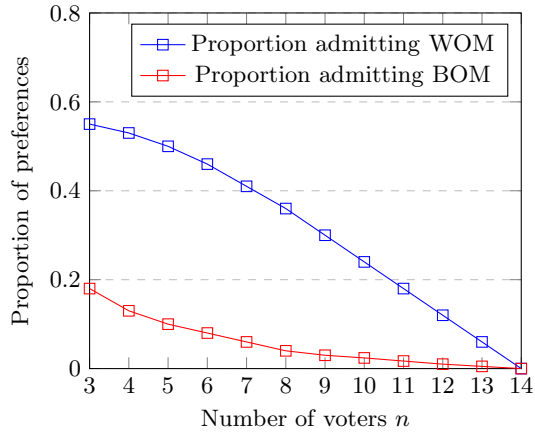
\end{center}

\newcommand\items{10}   
\newcommand{\STAB}[1]{\begin{tabular}{@{}c@{}}#1\end{tabular}}
\arrayrulecolor{white} 

\begin{figure}[h!]
\centering
\scalebox{0.91}{
 \noindent\begin{tabular}{cc*{\items}{|E}|}

 \multicolumn{11}{c}{Number of outcomes $m$}\\
 \multirow{10}{*}{\STAB{\rotatebox[origin=c]{90}{No. of disapprovals $m-k$}}} &
 \multicolumn{1}{c}{} &
 \multicolumn{1}{c}{21} &
 \multicolumn{1}{c}{22} &
 \multicolumn{1}{c}{23} &
 \multicolumn{1}{c}{24} &
 \multicolumn{1}{c}{25} &
 \multicolumn{1}{c}{26} &
 \multicolumn{1}{c}{27} &
 \multicolumn{1}{c}{28} &
 \multicolumn{1}{c}{29} &
 \multicolumn{1}{c}{30} \\  \hhline{~*\items{|-}|}

 & 1  & 0.61   & 0.61  & 0.62 & 0.62 & 0.63 & 0.63 & 0.64 & 0.64 & 0.65 & 0.65 \\ \hhline{~*\items{|-}|}
 & 2  & 0.80   & 0.81  & 0.82  & 0.83 & 0.84 & 0.84 & 0.85 & 0.85 & 0.86 & 0.86  \\ \hhline{~*\items{|-}|}
 & 3  & 0.85   & 0.87   & 0.88 & 0.89 & 0.90 & 0.90 & 0.91 & 0.92 & 0.92 & 0.93 \\ \hhline{~*\items{|-}|}
 & 4  & 0.83   & 0.86   & 0.88 & 0.89 & 0.91 & 0.92 & 0.93 & 0.94 & 0.94 & 0.95 \\ \hhline{~*\items{|-}|}
 & 5  & 0.74   & 0.80   & 0.83 & 0.86 & 0.89 & 0.90 & 0.92 & 0.93 & 0.94 & 0.95 \\ \hhline{~*\items{|-}|}
 & 6  & 0.47   & 0.60   & 0.70 & 0.77 & 0.82 & 0.85 & 0.88 & 0.90 & 0.92 & 0.93 \\ \hhline{~*\items{|-}|}
 & 7  & 0   & 0   & 0.28 & 0.48 & 0.61 & 0.71 & 0.78 & 0.83 & 0.87 & 0.89 \\ \hhline{~*\items{|-}|}
 & 8  & 0   & 0   & 0 & 0 & 0 & 0.29 & 0.49 & 0.62 & 0.72 & 0.79 \\ \hhline{~*\items{|-}|}
 & 9  & 0   & 0   & 0 & 0 & 0 & 0 & 0 & 0 & 0.29 & 0.50 \\ \hhline{~*\items{|-}|}
 \end{tabular}
}
%
\caption{Effect of $m$ and $m-k$ on proportion of OM-admitting preferences  ($n=3$)}\label{fig2}
\end{figure}
\section{Conclusion}
In this paper, we initiated research on the obvious manipulability of voting rules. One of our key insights is that certain rules are obviously manipulable when the number of outcomes is relatively large as compared to the number of voters. The $k$-approval rule is an example of such a rule, and we have also shown that under the rule, an obvious manipulation can be computed in polynomial time. Despite all unanimous, non-dictatorial voting rules being manipulable for $n\geq 3$, most commonly used rules are NOM, suggesting that NOM is a significantly weaker notion than strategyproofness. We remark that in the positional scoring rules that we have classified as OM, the obvious manipulations are dependent on a fixed, deterministic tiebreak ordering which is standard in the voting literature.
To gain further insights into which voting rules are more manipulable than others, a Bayesian approach could be used, in which voters have prior beliefs on the distribution of other votes. This approach lies between the perfect information of strategyproofness and the lack of information in NOM.
As a new concept, NOM has currently been examined only for a handful of settings. It will be interesting to consider it when analyzing the strategic behaviour of agents in other settings such as fair division (see, e.g., \citep{Orte19a}).

\noindent
\textbf{Acknowledgements.} The authors thanks Anton Baychkov, Barton Lee and the anonymous reviewers of ADT 2021 for useful feedback.
\vspace{-1em}
\bibliographystyle{splncs04nat}
 \newpage
 
 \normalsize
\section*{Appendix}

\paragraph{Proof of Lemma~\ref{prop:noveto}}

	\begin{proof}
		
We show that if a voting rule $f$ does not admit a non-dictatorial {voter with veto power}, then is not obviously manipulable.
Consider a voting rule $f$ that  {does not admit a voter with veto power}. The admission of a dictatorial voter implies that the rule is strategyproof, so it suffices to consider the case when $f$ is not dictatorial.

First note that voter $i$'s best outcome under truthful report $\succ_i$ is her most preferred outcome $o$ in the set of possible outcomes under $f$. Such an outcome is achievable because $o$ is a possible outcome and no voter has veto power, hence the outcome must be achievable under some $\succ_{-i}$. Therefore her best outcome under any untruthful report $\succ_i'$ cannot be strictly better than under a truthful report.

When $i$ reports $\succ_i'$, her worst possible outcome with respect to her preference $\succ_i$ is her least preferred outcome from the set of possible outcomes. Such an outcome is achievable because $f$ does not admit a {voter with veto power}.
We therefore have $W_{\succ_i}(\succ_i,f)\pref_i W_{\succ_i}(\succ_i',f)$ for all untruthful ballots $\succ_i'$. Therefore, $f$ is NOM. \qed
\end{proof}
\paragraph{Proof of Corollary~\ref{corol:NOM}}
			\begin{proof}
				Any majoritarian rule is almost-unanimous. Hence, the statement follows from Theorem~\ref{prop:almostunan}. \qed
				
				\end{proof}
				
\paragraph{Proof of Theorem~\ref{theo:kapp}}
\begin{proof}
The statement follows from Lemma \ref{k1} and Lemma \ref{k2}. \qed
\end{proof}

\paragraph{Proof of Corollary~\ref{corol:NOMM}}
\begin{proof}
	Note that for plurality, $k=1$. Hence, $m\leq n(m-k)+1$ holds.
\qed\end{proof}		
\paragraph{Proof of Lemma~\ref{lemma:spsr}}

\begin{proof}
Consider the scoring rule $w=(m+2,m+1,\ldots,4,0)$. Suppose we have $N=\{i,j,k\}$, $m=4$, $w=(6,5,4,0)$, tie-break order $o_1\succ_Lo_2\succ_Lo_3\succ_Lo_4$ and that voter $i$'s truthful ballot is:
\[\succ_i: o_1\succ_i o_2\succ_i o_3\succ_i o_4.\]
We first show that this scoring rule can admit a voter with veto power. Here, voter $i$ attempts to veto outcome $o_4$ by voting it last. We show that it is impossible for the other voters to vote such that outcome $o_4$ is chosen. Clearly, $j$ and $k$ must vote outcome $o_4$ as first preference, {giving it a score of 12}. Now outcome $o_1$ must have a strictly lower score than $o_4$, so we have three possible cases.\\
\textbf{Case 1:}
Outcome $o_1$ is set as the second preference of one voter and the last preference of the remaining voter.
\[\succ_j: o_4\succ_jo_1\succ_j\cdot\succ_j\cdot\]
\[\succ_k: o_4\succ_k \cdot\succ_k\cdot\succ_k o_1\]
Outcome $o_2$ must be ranked as $j$'s last preference, else it will have a strictly higher score than $o_4$. With $o_3$ as $j$'s third preference, it must have at least 12 points, and wins the tie-break with $o_4$. Therefore in this case there does not exist a voting profile that chooses outcome $o_4$.\\
\textbf{Case 2:}
Outcome $o_1$ is set as the third preference of one voter and the last preference of the remaining voter.
\[\succ_j: o_4\succ_j\cdot\succ_jo_1\succ_j\cdot\]
\[\succ_k: o_4\succ_k \cdot\succ_k\cdot\succ_k o_1\]
If outcome $o_2$ is ranked as the second preference of $j$ and either the second or third preference of $k$, it will have a strictly higher score than $o_4$. The same can be said for outcome $o_3$. Hence by exhaustion there does not exist a voting profile that chooses outcome $o_4$.\\
\textbf{Case 3:}
Outcome $o_1$ is set as the last preference of both remaining voters.
\[\succ_j: o_4\succ_j\cdot\succ_j\cdot\succ_j o_1\]
\[\succ_k: o_4\succ_k \cdot\succ_k\cdot\succ_k o_1\]
The lowest score that $o_2$ can achieve is 13, by setting it as the third preference of both $j$ and $k$. Therefore $o_4$ cannot be chosen.

By exhaustion of cases, we show that there does not exist a voting profile $\succ_{-i}$ that chooses outcome $o_4$, so it has been effectively vetoed by agent $i$.\\
\\
We now show that this voting rule is obviously manipulable. Suppose we have the preference profile 
\[\succ_i: o_1\succ_i o_2\succ_i o_4\succ_i o_3,\]
\[\succ_j: o_3\succ_jo_1\succ_jo_4\succ_jo_2,\]
\[\succ_k: o_3\succ_k o_2\succ_ko_4\succ_k o_1,\]
$W_{\succ_i}(\succ_i,f)=o_3$ is achievable as it wins the tie-break with outcome $o_4$. Now suppose that voter $i$ instead reports the manipulation
\[\succ'_i: o_1\succ_i o_4\succ_i o_2\succ_i o_3.\]
By a similar argument as above, it can be shown that it is impossible for the other voters to vote such that outcome $o_3$ is chosen. Therefore, $W_{\succ_i}(\succ'_i,f)\succ_iW_{\succ_i}(\succ_i,f)$, concluding our proof.\qed
\end{proof}


\paragraph{Proof of Lemma~\ref{yuck}}
\begin{proof}
Consider the scoring rule $w=(\omega+m\epsilon,\omega+(m-1)\epsilon,\ldots,\omega+2\epsilon,0)$, where $\omega,\epsilon>0$ and $\omega>m\epsilon(n-1)$. Suppose without loss of generality that voter $i$'s preferences are:
\[\succ_i:o_1\succ_io_2\succ_i\ldots\succ_io_m.\]
Since the voting rule is strict, the best case outcome is trivial.\\
\textbf{Case 1} ($m>n$):\\
We will show that this scoring rule admits a voter with veto power when the number of outcomes is greater than the number of voters. The highest score that outcome $o_m$ can receive is $(n-1)(\omega+m\epsilon)$. By the pigeonhole principle, there must exist at least one outcome which is not voted last preference by any voter. Each of these outcomes always has a score greater than $n\omega$, which is greater than $(n-1)(\omega+m\epsilon)$. It can therefore be seen that $o_m$ has been vetoed by voter $i$. Now suppose that the other $n-1$ voters report the following ballots:
\[o_{m-1}\succ\ldots\succ O\succ o_1\]
\[o_{m-1}\succ\ldots\succ O\succ o_2\]
\[\ldots\]
\[o_{m-1}\succ\ldots\succ O\succ o_{n-1},\]
where
\[O=\begin{cases}
o_m & m=n+1\\
o_{m-2} & m=n+2\\
o_{m-2}\succ\ldots\succ o_n & m>n+2.
\end{cases}\]
If $m=n+1$, then outcome $o_{m-1}$ trivially has the highest score and is chosen. We now consider $m>n+1$. Clearly, the two outcomes with the highest scores are $o_{m-1}$ and $o_{m-2}$. We now show that the score of $o_{m-1}$ is always greater than the score of $o_{m-2}$.
\begin{align*}
\mathrm{Score}(o_{m-1})&>\mathrm{Score}(o_{m-2})\\
\iff (\omega+2\epsilon)+(n-1)(\omega+m\epsilon)&>\omega+3\epsilon+(n-1)(\omega+(m-n-1)\epsilon)\\
\iff n\omega+\epsilon(2+mn-m)&>n\omega+\epsilon(4+mn-n^2-m)\\
\iff n^2&>4
\end{align*}
which always holds as we assume $n\geq 3$. Therefore, voter $i$'s worst case outcome under a truthful report is $W_{\succ_i}(\succ_i,w)=o_{m-1}$. Now any manipulation where $o_{m-1}$ increases in preference gives it a strictly higher score, so it remains chosen. Any manipulation where $o_{m-1}$ remains the same preference ranking can be reversed by the other voters making the same outcome changes in their ballots. Finally, any manipulation where $o_{m-1}$ moves to a lower preference ranking results in $o_m$ being unvetoed, making it the worst case outcome. Thus, this scoring rule is NOM when $m>n$.\\
\textbf{Case 2} ($m\leq n$):\\
Now suppose the other $n-1$ voters vote $o_m$ as their first preference, and that they vote such that every outcome is voted last preference by at least one voter. Excluding $o_m$, the outcome with the highest possible score is $o_1$, when it is voted second preference by $n-2$ of the other voters and last preference by 1 of the other voters. We now show that under this scenario, $o_m$ has a higher score than $o_1$.
\begin{align*}
\mathrm{Score}(o_m)&>\mathrm{Score}(o_1)\\
\iff (n-1)(\omega+m\epsilon)&>(\omega+m\epsilon)+(n-2)(\omega+(m-1)\epsilon)\\
\iff \epsilon (mn-n)&>\epsilon(mn-m-n+2)\\
\iff n&>2
\end{align*}
which always holds as we assume $n\geq 3$. $o_m$ is therefore voter $i$'s worst case outcome. Using the same arguments as in Case 1, we deduce that $o_m$ is also the worst case outcome under any manipulation by voter $i$. This scoring rule is therefore NOM under $m\leq n$. \qed
\end{proof}


\paragraph{Proof of Theorem~\ref{th:sPOS}}
\begin{proof}
Suppose we have weight vector $w=(s_1,s_2,\ldots,s_m)$, where $s_i-s_{i+1}\geq s_{i+1}-s_{i+2}$ for all $i\in\{1,2,\ldots,m-2\}$. The rule is strict, so it is not best case obviously manipulable because $i$'s most preferred outcome gets selected if it is reported to be in the first position by all the voters. 

Next, we show that that rule is not worst case obviously manipulable. We do so by showing that whenever the voter $i$ misreports, there exists a profile of the other voters under which $i$'s least preferred outcome is selected. 

Consider the scenario where voter $i$'s truthful ballot is
\[\succ_i: o_1\succ_io_2\succ_i\ldots\succ_io_m,\]
and every other voter reports the reverse preference order
\[\succ_{-i}:o_m\succ_{-i}o_{m-1}\succ_{-i}\ldots\succ_{-i}o_2\succ_{-i}o_1.\]
Alternative $o_m$ has a score of $s_m+(n-1)s_1$, and for $i\in \{1,2,\ldots,m-1\}$, $o_i$ has a score of $s_i+(n-1)s_{m+1-i}$. We show that $s_m+(n-1)s_1>s_i+(n-1)s_{m+1-i}$ for $i\in \{1,2,\ldots,m-1\}$.
\begin{align*}
&s_m+(n-1)s_1>s_i+(n-1)s_{m+1-i}\\
\iff&(n-1)(s_1-s_{m+1-i})>s_i-s_m
\end{align*}
Since $n\geq 3$, the inequality holds for $i=1$. It suffices to show that $s_1-s_{m+1-i}\geq s_i-s_m$ for all $i\in \{2,3,\ldots,m-1\}$. Now due to weakly diminishing differences, the following inequalities hold:
\begin{align*}
s_1-s_2&\geq s_i-s_{i+1}\\
s_2-s_3&\geq s_{i+1}-s_{i+2}\\
&\cdots\\
s_{m-1-i}-s_{m-i}&\geq s_{m-2}-s_{m-1}\\
s_{m-i}-s_{m+1-i}&\geq s_{m-1}-s_m.
\end{align*}
If we sum up these inequalities, we have $s_1-s_{m+1-i}\geq s_i-s_m$, so therefore $s_m+(n-1)s_1>s_i+(n-1)s_{m+1-i}$ for $i\in \{1,2,\ldots,m-1\}$, meaning outcome $o_m$ has a strictly higher score than every other outcome. A manipulation by $i$ where $o_m$ is not her last preference results in the outcome having a strictly higher score and remaining chosen by every other agent reporting the reverse preference order and swapping $o_m$ to become their first preference. If $i$ reports the manipulation $\succ'_i$ such that $o_m$ is still her last preference, it can be shown by the same argument that $o_m$ remains her worst case outcome, achieved by the other voters reporting the reverse preference order. Therefore the rule is not worst case obviously manipulable. \qed
\end{proof}

\paragraph{Proof of Lemma~\ref{BOM}}
\begin{proof}
Suppose voter $i$ has preferences
\[\succ_i: o_1\succ_i o_2\succ_i\ldots\succ_io_m.\]
Consider the instance where every voter reports the same preferences as $i$. If the first and second elements of the scoring vector are different, then $B_{\succ_i}(\succ_i,f)=o_1$ is achieved, which cannot be manipulated further. Therefore, best case obvious manipulability requires the first $k$ elements of the scoring vector to be the same for some $k>1$.

Now suppose the first $k$ elements of the scoring vector are equal for some $k>1$. First, we show that if $k-1\leq (n-1)(m-k)$, then $f$ is not best case obviously manipulable.

\noindent \textbf{Case 1 ($k-1\leq (n-1)(m-k)$):}
We show by construction that if this condition is met, there always exists a profile where $i$ reports truthfully and $o_1$ is the unique outcome with the highest score. Under the profile where every voter reports the same preferences as $i$, outcomes $\{o_1,\dots,o_k\}$ have a tied score. A voter can ensure an outcome of $\{o_2,\dots,o_k\}$ has a strictly lower score than $o_1$ by `swapping' its preference relation positioning with one of its $m-k$ least preferred outcomes. For example, the preference relation $o_1\succ o_2\succ \dots \succ o_m$ changes to $o_1\succ o_m \succ o_3 \succ \dots \succ o_{m-1}\succ o_2$. There are $n-1$ voters that can perform $m-k$ `swaps', so if $k-1\leq (n-1)(m-k)$, there are sufficient swaps to ensure each outcome in $\{o_2,\dots,o_k\}$ has a strictly lower score than $o_1$. Therefore, there always exists a profile where $i$ reports truthfully and $o_1$ is the unique outcome with the highest score. A best case outcome of $B_{\succ_i}(\succ_i,f)=o_1$ cannot be improved, so if $k-1\leq (n-1)(m-k)$, then $f$ is not best case obviously manipulable.

\noindent \textbf{Case 2 ($k-1> (n-1)(m-k)$):}
We show that if this condition is met, there exists a tie-breaking ordering and manipulation $\succ_i'$ where $B_{\succ_i}(\succ'_i,f) \succ_i B_{\succ_i}(\succ_i,f)$. Consider the same scenario as in the previous case. If $k-1> (n-1)(m-k)$, then there are insufficient swaps to ensure that each outcome in $\{o_2,\dots,o_k\}$ has a strictly lower score than $o_1$. Consequently, under any profile where $i$ reports truthfully and all voters vote $o_1$ as their first preference, there must be at least one outcome in $\{o_2,\dots,o_k\}$ that is tied with $o_1$. Denote $c=(k-1)-(n-1)(m-k)$ as the minimum number of outcomes in $\{o_2,\dots,o_k\}$ that must be tied with $o_1$ and suppose that the tie-breaking order is \[\succ_L: o_k\succ_Lo_{k-1}\succ_L\ldots\succ_Lo_2\succ_Lo_1\succ_Lo_{k+1}\succ_L\dots\succ_L o_m.\]
By iteratively selecting $i$'s $c$ most preferred outcomes in $\{o_2,\dots,o_{k}\}$ to be the set of outcomes tied with $o_1$, we deduce that $B_{\succ_i}(\succ_i,f)=o_{c+1}$, as it is the tiebreak winner out of the set of tied outcomes $\{o_1,\dots,o_{c+1}\}$. Now suppose that voter $i$ instead reports the manipulation where $o_{c+1}$ has been swapped with $o_m$
\[\succ_i': o_1\succ_i\dots\succ_io_m\succ_i\dots\succ_io_k\dots\succ_io_{c+1}.\]
The other voters can now report such that the set of outcomes $\{o_2,\dots,o_c,o_m\}$ is tied with $o_1$. We now have $B_{\succ_i}(\succ'_i,f)=o_c$ and therefore $B_{\succ_i}(\succ'_i,f) \succ_i B_{\succ_i}(\succ_i,f)$.

Therefore a positional scoring rule $f$ where the first $k$ elements of the scoring vector are equal for some $k>1$ is best case obviously manipulable if and only if $k-1> (n-1)(m-k)$ (or $k-2\geq(n-1)(m-k)$), which can be rearranged to form $n \leq \frac{m-2}{m-k}$. The lemma's statement follows from the fact that best case obvious manipulability requires the first $k$ elements of the scoring vector to be the same for some $k>1$. \qed
\end{proof}
\begin{theorem}
 A positional scoring rule is NOM if the tiebreak is randomized and each tied outcome has a positive probability of winning the tiebreak.
 \end{theorem}
 \begin{proof}
 Suppose there are $n$ voters, $m$ outcomes, the weight vector of the scoring rule is $w=(s_1,s_2,\dots,s_m)$, and without loss of generality that voter $i$'s true preferences are $\succ_i: o_1\succ_i o_2\succ_i \dots \succ_i o_m$.

 We first show that a positional scoring rule $f$ is NOM in the best case if the condition is met. Recall that our definition of a positional scoring rule ensures that $s_1\geq s_2 \geq \dots \geq o_m$. Suppose every other voter votes the same ballot as voter $i$. If $s_1>s_2$, $B_{\succ_i}(\succ_i,f)=o_1$ and cannot be improved by any manipulation. Now if $s_1=s_2=\dots=s_k$ for some $k\in \{2,\dots,m-1\}$, then outcomes $o_1,\dots,o_k$ are tied with the highest score. By the tiebreak condition, $o_1$ has a positive probability of winning the tiebreak, so we have $B_{\succ_i}(\succ_i,f)=o_1$ which cannot be improved by any manipulation. The positional scoring rule is therefore NOM in the best case.

 We now show that it is NOM in the worst case if the condition is met. Denote $o_w=W_{\succ_i}(\succ_i,f)$ and let $\succ_{-i}$ be the preference profile such that $f(\succ_i,\succ_{-i})=o_w$. Under the preference profile $(\succ_i,\succ_{-i})$ and scoring rule $f$, the expression $Score(o_w)\geq \max_{j\in \{1,\dots,m\}}Score(o_j)$ holds. Note that under this scenario, $o_w$ can have the same score as any other outcome due to the tiebreak condition.

 Now suppose voter $i$ misreports $\succ_i'$ where $o_w$ is in the same relative position as in $\succ_i$. Define $\sigma$ as the function that relabels/permutes the outcomes such that $\sigma(\succ_i)=\succ_i'$. If the other voters apply the same relabelling/permutation to their own ballots and vote preference profile $\succ_{-i}'=(\sigma(\succ_1), \dots, \sigma(\succ_{i-1}), \sigma(\succ_{i+1}), \dots, \sigma(\succ_n))$, then we still have $Score(o_w)\geq \max_{j\in \{1,\dots,m\}}Score(o_j)$ and hence $o_w$ remains as voter $i$'s worst case outcome.

 Consider the case where voter $i$ misreports $\succ_i^*$ where $o_w$ is more preferred relative to $\succ_i$. We construct $\succ_i'$ from $\succ_i^*$ by swapping $o_w$ with the outcome in its original position in $\succ_i$, and let $\sigma$ be such that $\sigma(\succ_i)=\succ_i'$. As previously shown, if the other voters vote $\succ_{-i}'$, then under $(\succ_i',\succ_{-i}')$, $Score(o_w)\geq \max_{j\in \{1,\dots,m\}}Score(o_j)$. If we convert $(\succ_i',\succ_{-i}')$ to $(\succ_i^*, \succ_{-i}')$ by making only one swap where $o_w$ moves to a higher position, the score of $o_w$ does not decrease, and no other outcome's score increases, so the inequality is retained and $o_w$ is still able to be chosen under $\succ_i^*$.

 Finally, we consider the case where voter $i$ misreports $\succ_i^\dagger$ where $o_w$ is less preferred relative to $\succ_i$. Similarly to the previous case, we construct $\succ_i'$ from $\succ_i^\dagger$ by swapping $o_w$ such that it is in the same relative position as in $\succ_i$, and define $\sigma$ such that $\sigma(\succ_i)=\succ_i'$. Under $(\succ_i',\succ_{-i}')$, we have $Score(o_w)\geq \max_{j\in \{1,\dots,m\}}Score(o_j)$. If we convert $(\succ_i',\succ_{-i}')$ to $(\succ_i^\dagger, \succ_{-i}')$ by making only one swap where $o_w$ moves to a lower position, then the score of $o_w$ does not increase, and the score of the other swapped outcome does not decrease. Under this profile, either $o_w$ or the other swapped outcome (which voter $i$ prefers less than $o_w$) is chosen by $f$, and hence voter $i$ does not improve its worst case outcome.

 By exhaustion of cases, we have shown that under the specified randomized tiebreak condition, any positional scoring rule is NOM in the worst case. The theorem statement follows. \qed
 \end{proof}

\end{document}